\documentclass[runningheads]{llncs} 
\usepackage{amssymb}
\usepackage[T1]{fontenc}
\usepackage[utf8]{inputenc}
\usepackage{mathtools}
\usepackage{xfrac}
\usepackage{braket}
\usepackage{physics}
\usepackage[inline]{enumitem}
\usepackage{xcolor}
\usepackage[outline]{contour}
\usepackage{cleveref}
\usepackage{lmodern}

\providecommand\given{} % so it exists
\DeclarePairedDelimiterX\PH[1][]{
   \renewcommand\given{\nonscript\:\delimsize\vert\nonscript\:}
   #1
}
\newcommand\Prob{\operatorname{Pr}\PH}

\DeclareRobustCommand
  \Compactdotsc{\mathinner{.\mkern-2mu.\mkern-2mu.}}

\renewcommand{\dotsc}{\Compactdotsc}

\usepackage{ulem}
\usepackage{tikz}
\usetikzlibrary{automata, positioning, decorations.pathmorphing, calc}

\tikzset{
	loop right/.append style={
		every loop/.append style={
			out=30, in=-30, looseness=4,
			shorten >= 0pt
		}
	},
	loop above/.append style={
		every loop/.append style={
			out=120, in=60, looseness=4,
			shorten >= 0pt
		}
	},
	loop left/.append style={
		every loop/.append style={
			out=210, in=150, looseness=4,
			shorten >= 0pt
		}
	},
	loop below/.append style={
		every loop/.append style={
			out=300, in=240, looseness=4,
			shorten >= 0pt
		}
	},
	every edge/.append style={
		every node/.append style={
			execute at begin node=$,
			execute at end node=$
		}
	},
	every state/.append style={
		execute at begin node=$,
		execute at end node=$
	},
	>=stealth,
	node distance = 1.65cm and 4cm,
	on grid,
	auto,
	initial text = 
}

\newcommand\mydefn[1]{\textcolor{red!60!black}{\itshape #1}}

\newcommand{\CC}[1]{\textsf{#1}}

\newcommand{\eL}{\mathcal{L}}
\renewcommand{\L}[1]{\eL\left(#1\right)}
\newcommand{\NL}{\CC{NL}}
\newcommand{\rNL}{\CC{rNL}}

\newcommand{\DFA}{\textsf{DFA}}

\newcommand{\onfa}{\textsf{1nfa}}
\newcommand{\odfa}{\textsf{1dfa}}
\newcommand{\tonfa}[1]{\textsf{\onfa\ensuremath{\left(#1\right)}}}
\newcommand{\todfa}[1]{\textsf{\odfa\ensuremath{\left(#1\right)}}}

\newcommand{\nfa}{\textsf{2nfa}}
\newcommand{\dfa}{\textsf{2dfa}}
\newcommand{\tnfa}[1]{\textsf{\nfa\ensuremath{\left(#1\right)}}}
\newcommand{\tdfa}[1]{\textsf{\dfa\ensuremath{\left(#1\right)}}}

\newcommand{\pfa}{\textsf{2pfa}}
\newcommand{\tpfa}[1]{\textsf{\pfa\ensuremath{\left(#1\right)}}}

\newcommand{\owips}[1]{\textsf{IP\ensuremath{\left(#1\right)}}}
\newcommand{\specs}{\textsf{restriction-list}}
\newcommand{\constr}{\textsf{constant-randomness}}

\DeclareMathOperator{\SUM}{sum}
\DeclarePairedDelimiter{\ceil}{\lceil}{\rceil}

\newcommand{\No}{\mathbb{N}_1}

\renewcommand{\P}{\mathcal{P}}

\newcommand{\Accept}{\textit{Accept}}
\newcommand{\Reject}{\textit{Reject}}

\newcommand{\reject}{\textit{reject}}

\newcommand{\textiff}{\textit{iff}}

\newcommand{\lend}{\ensuremath{\triangleright}}
\newcommand{\rend}{\ensuremath{\triangleleft}}

\newcommand{\trim}[1]{#1^\times}

\newcommand{\err}{\varepsilon}
\newcommand{\faerr}{\err_{\textrm{false-accept}}}
\newcommand{\failaerr}{\err_{\textrm{fail-to-accept}}}
\newcommand{\failrerr}{\err_{\textrm{fail-to-reject}}}
\newcommand{\weaklabel}{\textrm{weak}}
\newcommand{\stronglabel}{\textrm{strong}}

\newcommand{\werr}{\err_{\weaklabel}}
\newcommand{\serr}{\err_{\stronglabel}}
\newcommand{\Lw}[2]{\eL_{\weaklabel, #1}\left(#2\right)}
\newcommand{\Ls}[2]{\eL_{\stronglabel, #1}\left(#2\right)}

\newcommand{\windablelabel}{\textrm{W}}
\newcommand{\reliablelabel}{\textrm{R}}
\newcommand{\kw}{k_{\windablelabel}}
\newcommand{\kr}{k_{\reliablelabel}}

\newcommand\overallp{P}
\newcommand\opr{\overallp_{\reliablelabel}}
\newcommand\opw{\overallp_{\windablelabel}}
\newcommand\individualp{p}
\newcommand\pr{\individualp_{\reliablelabel}}
\newcommand\pw{\individualp_{\windablelabel}}
\newcommand\Propt{\opr^*}

\newcommand\FR{F}
\newcommand\FRr{\FR_{\reliablelabel}}
\newcommand\FRw{\FR_{\windablelabel}}

\newcommand{\ta}{\mathtt{a}}
\newcommand{\tb}{\mathtt{b}}
\newcommand{\tc}{\mathtt{c}}
\newcommand{\td}{\mathtt{d}}

\newcommand{\tsh}{\mathtt{\#}}
\newcommand{\ext}{\mathop{\vcenter{\hbox{$\scriptstyle\!\circ$}}}}
\newcommand{\wo}[2]{#1_{\setminus #2}}

\DeclarePairedDelimiter\br{\lbrack}{\rbrack}
\newcommand{\BR}[1]{\br[\big]{#1}}

\setlist{parsep=0.1em, itemsep=0.3em}

\newcommand\POS\rho
\newcommand\D{t}

\newcommand\C{m}
\newcommand\B{b}

\newcommand\openrange[2]{\left(#1, #2 \right)}

\title{
	Windable Heads \& \\
	Recognizing \NL{} with Constant Randomness
}

\author{Mehmet Utkan Gezer\orcidID{0000-0002-5022-178X}}
\authorrunning{M. U. Gezer}
\institute{Boğaziçi University, Bebek İstanbul, Türkiye\\
\email{utkan.gezer@boun.edu.tr}}

\begin{document}
% \fontsize{10pt}{13.5pt}
% \selectfont

\maketitle

\begin{abstract}
	Every language in \NL{} has a $k$-head two-way
	nondeterministic finite automaton (\tnfa{k}) recognizing it.
	It is known how to build a constant-space verifier algorithm
	from a \tnfa{k} for the same language with constant-randomness,
	but with error probability $\sfrac{k^2-1}{2k^2}$
	that can not be reduced further by repetition.
	We have defined the unpleasant characteristic
	of the heads that causes the high error as
	the property of being ``windable''.
	With a tweak on the previous verification algorithm,
	the error is improved to $\sfrac{\kw^2-1}{2\kw^2}$,
	where $\kw \le k$ is the number of windable heads.
	Using this new algorithm,
	a subset of languages in \NL{} that have
	a \tnfa{k} recognizer
	with $\kw \le 1$ can be verified
	with arbitrarily reducible error using
	constant space and randomness. 

\end{abstract}

\section{Introduction}

Probabilistic Turing machines (PTM) are
otherwise deterministic Turing machines
with randomness as a resource.
They can be standalone recognizers of languages,
or be verifiers for the proofs of memberships.
In either scenario,
a measurable error is incorporated into their decision
due to randomness involved in their execution.
This error can usually be reduced
via repeated execution in the PTM's control.

The language class verifiable by the
constant-randomness two-way probabilistic
finite automata (\pfa) is the class \NL{}.
Curiously, however, the error of these verifiers
in recognizing languages of this class
seems to be irreducible beyond a certain threshold~\cite{say_finite_2014}.

In this paper,
we introduce a characteristic for the languages in \NL{}.
Based on this characteristic, we lower the error threshold
established in~\cite{say_finite_2014}
for almost all languages in \NL{}.
Finally, we delineate a subset of \NL{}
which are verifiable by the constant-randomness \pfa{}
with arbitrarily low error.

The remaining of the paper is structured as follows:
\Cref{sec:kfa,sec:ips} provides the necessary background
as well as our terminology in the domain.
A key property of the multi-head finite automata
is identified in \cref{sec:wind}.
The characterization of languages in \NL{} and
our algorithm for verification achieving aforementioned results
are described in \cref{sec:alg}.

Following notation will be common throughout this paper:
\begin{itemize}
	\item
		$\L{M}$ denotes the language recognized by the machine $M$.
	\item
		$\L{\CC{X}} = \Set{ \L{M} | M \in \CC{X} }$
		for a class of machines $\CC{X}$.
	\item
		$\wo{S}{q}$ denotes the set $S$ without its element $q$.
	\item
		$\sigma_i$ denotes the $i$th element of the sequence $\sigma$.
	\item
		$\trim{w}$ denotes the substring of $w$ without its last character.
	\item
		$\sigma\ext\tau$ denotes the sequence $\sigma$
		concatenated with the element or sequence $\tau$.
\end{itemize}

\section{Finite automata with $k$ heads}\label{sec:kfa}

Finite automata are the Turing machines
with read-only tape heads on a single tape.
A finite automata with only one head is
equivalent to a \DFA{} (deterministic finite automaton)
in terms of language recognition~\cite{holzer_complexity_2011},
hence recognizes a regular language.
Finite automata with $k > 1$ heads
can recognize more than just regular languages.
Their formal definition may be given as follows:

\begin{definition}[Multi-head nondeterministic finite automata]\label{def:tnfa}
A \tnfa{k} is a 5-tuple,
$M = (Q, \Sigma, \delta, q_0, q_f)$, where;
\begin{enumerate}
	\item $Q$ is the finite set of states,
	\item $\Sigma$ is the finite set of input symbols,
		\begin{enumerate}
			\item 
				$\lend, \rend$ are the left and right
				end-markers for the input on the tape,
			\item
				$\Gamma = \Sigma \cup \Set{ \lend, \rend }$
				is the tape alphabet,
		\end{enumerate}
	\item
		$\delta \colon Q \times \Gamma^{k} \to
		\P(\wo{Q}{q_0} \times \Delta^{k})$
		is the transition function, where;
		\begin{enumerate}
			\item
				$\Delta = \Set{ -1, 0, 1}$
				is the set of head movements,
		\end{enumerate}
	\item $q_0 \in Q$ is the unique initial state,
	\item $q_f \in Q$ is the unique accepting state.
\end{enumerate}
\end{definition}

Machine $M$ is said to execute on a string $w \in \Sigma^*$,
when $\lend w \rend$ is written onto $M$'s tape,
all of its heads rewound to the cell with $\lend$,
its state is reset to $q_0$,
and then it executes in steps by the rules of $\delta$.
At each step,
inputs to $\delta$ are the state of $M$
and the symbols read by respective heads of $M$.

When $\abs{\delta} = 1$ with the only member
$(q', (d_1, \dotsc, d_k)) \in \wo{Q}{q_0} \times \Delta^k$,
the next state of $M$ becomes $q'$,
and $M$ moves its $i$th head by $d_i$.
Whenever $\abs{\delta} > 1$, the execution branches,
and each branch runs in parallel.
A branch is said to reject $w$, if $\abs{\delta} = 0$,
or if all of its branches reject.
A branch accepts $w$, if its state is at $q_f$,
or if any one of its branches accepts.
A branch may also do neither, in which case
the branch is said to loop.

A string $w$ is in $\L{M}$, if the root of $M$'s
execution on $w$ is an accepting branch.
Otherwise, $w \notin \L{M}$,
and the root of $M$'s execution is either a rejecting or
a looping branch.

Restricting $\delta$ to not have transitions inbound to $q_0$
does not detriment the language recognition of a \tnfa{k}
in terms of its language recognition:
Any \tnfa{k} with such transitions can be converted into one without,
by adding a new initial state $q'_0$
and setting $\delta(q'_0, \lend, \dotsc, \lend) = \Set{ (q_0, 0, \dotsc, 0) }$.

\begin{lemma}
	The containment
	$\L{\tnfa{k}} \subsetneq \L{\tnfa{k+1}}$
	is proper~\cite{monien_transformational_1976,monien_two-way_1980,sudborough_remarks_1977,seiferas_relating_1977,ibarra_two-way_1973}.
\end{lemma}

\begin{lemma}\label{lem:timelimit}
	There is a way to obtain a \tnfa{2k}
	that is guaranteed to halt,
	from any given \tnfa{k}.
\end{lemma}

\begin{proof}
	A $k$-headed automaton running on an input $w$ of length $n$
	has $n^k$ distinct configurations.
	Additional $k$ heads can count up to $n^k = (nnn\dots n)_n$,
	and halt the machine with a rejection.
\end{proof}

\begin{lemma}\label{lem:nomovebeyond}
	Every \tnfa{k} can be converted into an equivalent \tnfa{k}
	which does not move its heads beyond the end markers.
\end{lemma}
This is done via trivial modifications on the transition function.

\begin{definition}[Multi-head deterministic finite automata]
	A \tdfa{k} is a \tnfa{k} that is restricted to satisfy
	$\abs{\delta} \le 1$, where $\delta$ is its transition function.
\end{definition}

\begin{lemma}\label{lem:knfaisnl}
	Following are shown in~\cite{hartmanis_non-determinancy_1972}:
	\begin{align}
		\cup_{k=1}^\infty \L{\tnfa{k}} &= \NL\\
		\cup_{k=1}^\infty \L{\tdfa{k}} &= \CC{L}
	\end{align}
\end{lemma}

\begin{definition}[Multi-head one-way finite automata]
	A \tonfa{k} is a restricted \tnfa{k} that does
	not move its heads backwards on the tape.
	In its definition, $\Delta = \Set{0, 1}$.
	A \todfa{k} is similarly a restriction of \tdfa{k}.
\end{definition}

\begin{definition}[Multi-head probabilistic finite automata]
	A \tpfa{k} $M$ is a PTM defined similar to a \tnfa{k}
	with the following modifications on \cref{def:tnfa}:
	\begin{enumerate}
		\item[1.${}'$]
			$Q = Q_D \cup Q_P$, where $Q_D$ and $Q_P$ are disjoint.
		\item[3.${}'$]
			Transition function $\delta$ is overloaded as follows:
			\begin{itemize}
				\item
					$\delta \colon Q_D \times \Gamma^{k} \to
					\P(\wo{Q}{q_0} \times \Delta^{k})$
				\item
					$\delta \colon Q_P \times \Gamma^{k}
					\times \Set{0,1} \to
					\P(\wo{Q}{q_0} \times \Delta^{k})$
			\end{itemize}
		The output of $\delta$ may at most have 1 element.
	\end{enumerate}
\end{definition}

States $Q_D$ are called deterministic,
and $Q_P$ probabilistic.
Depending on the state of the machine,
$\delta$ receives a third parameter,
where a 0 or 1 is provided by
a random bit-stream.

A string $w$ is in $\L{M}$,
if and only if $M$ accepts $w$ with a probability greater than $\sfrac12$.

Due to the probabilistic nature of a given \tpfa{k} $M$,
following three measures of error in the language
recognition are inherent to it:
\begin{align}
	\tag{False rejection}
	\failaerr(M) &= \Prob{M \text{ does not accept } w \given w \in \L{M}}\\
	\tag{Failure to reject}
	\failrerr(M) &= \Prob{M \text{ does not reject } w \given w \notin \L{M}}\\
	\tag{False acceptance}
	\faerr(M) &= \Prob{M \text{ accepts } w \given w \notin \L{M}}
\end{align}

Note that when a \tpfa{k} $M$ does not reject a string $w$,
then it could have either accepted it,
or wound up in an infinite loop.
Consequently, $\failrerr \ge \faerr$ is always true.
Based on this fact, the overall weak and strong
errors of a probabilistic machine $M$ is defined as follows:
\begin{align}
	\tag{Weak error}
	\werr(M) &= \max(\failaerr(M), \faerr(M))\\
	\tag{Strong error}
	\serr(M) &= \max(\failaerr(M), \failrerr(M))
\end{align}

Given a $k$ and $\err < \sfrac12$, let
\begin{equation*}
	\Lw{\err}{\tpfa{k}} = \Set{ \L{M} | M \in \tpfa{k}, \werr(M) \le \err }
\end{equation*}
be the class of languages recognized
by a \tpfa{k} with a weak error at most $\err$.
Class $\Ls{\err}{\tpfa{k}}$ is defined similarly.

\section{Interactive Proof Systems}\label{sec:ips}

An interactive proof system (IPS) models the verification
process of proofs.
Of the two components in an IPS, the \emph{prover}
produces the purported proof of membership for a given input string,
while the \emph{verifier}
either accepts or rejects the string, alongside its proof.
The catch is that the prover is assumed to advocate for the
input string's membership without regards to truth,
and the verifier is expected to be accurate in its decision,
holding a healthy level of skepticism against the proof.

The verifier is any Turing machine with capabilities
to interact with the prover via a shared communication cell.
The prover can be seen as an infinite state transducer
that has access to both
an original copy of the input string
and the communication cell.
Prover never halts, and its output is to the communication cell.

Our focus will be on the one-way IPS,
which restricts the interaction to be a monologue
from the prover to the verifier.
Since there is no influx of information to the prover,
prover's output will be dependent on the input string only.
Consequently, a one-way IPS can also be modeled
as a verifier paired with a certificate function,
$c \colon \Sigma^* \to \Lambda^\infty$,
where $\Lambda$ is the communication alphabet.
% https://en.wikipedia.org/wiki/Omega_language
A formal definition follows:

\begin{definition}[One-way interactive proof systems]
	An \owips{\specs} is defined with a tuple of
	a verifier and a certificate function, $S = (V, c)$.
	The verifier $V$ is a Turing machine
	of type specified by the \specs{}.
	The certificate function $c$
	outputs the claimed proof of membership
	$c(w) \in \Lambda^\infty$ for a given input string $w$.
\end{definition}

The verifier's access to the certificate is only
in the forward direction.
The qualifier ``one-way'', however,
specifies that the interaction in the IPS
is a monologue from the prover to the verifier,
not the aforementioned fact, which is true for all IPS.

The language recognized by $S$ can be denoted with $\L{S}$,
as well as $\L{V}$.
A string $w$ is in $\L{S}$,
if and only if the interaction results in an acceptance of $w$ by $V$.

If the verifier of the IPS is probabilistic,
its error becomes the error of the IPS.
The notation $\Lw{\err}{\owips{\specs}}$ and
$\Ls{\err}{\owips{\specs}}$ is also adopted.

Say and Yakaryılmaz proved that~\cite{say_finite_2014}:
\begin{align}
	\NL &\subseteq \Lw{\err}{\owips{\tpfa{1}, \constr}}
		&& \text{for $\err > 0$ arbitrarily small,}
		\label{eq:weak}\\
	\NL &\subseteq \Ls{\err}{\owips{\tpfa{1}, \constr}}
		&& \text{for $\err = \frac12 - \frac1{2k^2}$, $k\to\infty$.}
		\label{eq:strong}
\end{align}

For the latter proposition, the research proves that
any language $L \in \NL$
can be recognized by a one-way IPS
$S \in \owips{\tpfa{1}, \constr}$,
which satisfies $\serr(S) \le \sfrac12 - \sfrac1{2k}$,
and where $k$ is the minimum number of heads
among the $\tnfa{k}$ recognizing $L$
that also halts on every input.
Existence of such a \tnfa{k} is guaranteed by
\cref{lem:knfaisnl,lem:timelimit}.

This work improves on the findings of~\cite{say_finite_2014}.
For their pertinence, an outline of the algorithms
attaining the errors in \cref{eq:weak,eq:strong} is provided
in the following sections.

\subsection{Reducing weak error arbitrarily using constant-randomness verifier}

Given a language $L \in \NL$ with a halting \tnfa{k} recognizer $M$,
verifier $V_1 \in \tpfa{1}$ expects a certificate to report
\begin{enumerate*}[label=(\roman*)]
	\item
		the $k$ symbols read, and
	\item
		the nondeterministic branch taken
\end{enumerate*}
for each transition made by $M$ on the course of accepting $w$.
Such a report necessarily contains a lie, if $w \notin \L{M} = L$.

Verifier $V_1$ has an internal representation of $M$'s control.
Then, the algorithm for the verifier is as follows:
\begin{enumerate}
	\item
		Repeat $\C$ times:
		\begin{enumerate}
			\item
				Move head left, until $\lend$ is read.
			\item
				Reset $M$'s state in the internal representation, denoted $q_m$.
			\item\label{itm:randchoice}
				Randomly choose a head of $M$ by flipping $\ceil{\log k}$ coins.
			\item
				Repeat until $q_m$ becomes the accepting state of $M$:
				\begin{enumerate}
					\item
						Read $k$ symbols and the nondeterministic
						branch taken by $M$ from the certificate.
					\item
						\Reject{} if the reading from $V_1$'s head disagrees with
						the corresponding symbol on the certificate.
					\item
						Make the transition in the internal representation
						if it is valid, and move the chosen head
						as dictated by the nondeterministic branch.
						\Reject{} otherwise.
				\end{enumerate}
		\end{enumerate}
	\item
		\Accept.
\end{enumerate}

For the worst case errors,
it is assumed that there is a lie for the certificate to tell
about each one of the heads alone
and in any single one of the transitions,
which causes $V_1$ to fail to reject a string $w \notin L$.
Similar lies are assumed to exist for the false acceptances.
Following are then the (upper bounds of) errors for $V_1$:
\begin{align*}
	\failaerr(V_1) &= 0 &
	\failrerr(V_1) &\le \frac{k-1}{k} &
	\faerr(V_1) &\le \frac{1}{k^m}
\end{align*}

A discrepancy between $\faerr$ and $\failrerr$
is observed, because an adversarial certificate
may wind $V_1$ up in an infinite loop on its first round
of $m$ repetitions.
This is possible despite $M$ being a halting machine.
The lie in the certificate can present an infinite
and even changing input string from the perspective
of the head being lied about.

Being wound up counts as a failure to reject,
but does not yield a false acceptance.
The resulting weak error is $\serr = k^{-m}$,
which can be made arbitrarily small.

\subsection{Bringing strong error below \sfrac12 using constant-randomness
verifier}

Presented first in~\cite{say_finite_2014},
verifier $V_1'$ with the following algorithm manages to
achieve $\serr(V_1') < \sfrac12$,
outlined as follows:
\begin{enumerate}
	\item
		Randomly \reject{} with $\sfrac{k-1}{2k}$ probability
		by flipping $\ceil{\log k} + 1$ coins.
	\item
		Continue as $V_1$.
\end{enumerate}

This algorithm then has the following upper bounds for the errors:
\begin{align*}
	\failaerr(V_1') &= \frac{k-1}{2k} &
	\failrerr(V_1') &\le \frac{k^2-1}{2k^2} &
	\faerr(V_1') &\le \frac{k+1}{2k^{m+1}}
\end{align*}

Since $\failrerr(V_1')$ is potentially greater than $\failaerr(V_1')$,
the strong error is bounded by $\sfrac{k^2-1}{2k^2}$.

\section{Windable heads}\label{sec:wind}

This section will introduce a property of the heads of a \tnfa{k}.
It leads to a characterization of the \tnfa{k}
by the number of heads with this property.
A subset \rNL{} of the class \NL{} will be defined,
which will also be a subset of $\Ls{\err}{\owips{\tpfa{1}, \constr}}$
for $\err > 0$ approaching zero.

A head of a \tnfa{k} $M$ is said to be \mydefn{windable}
if these three conditions hold:
\begin{itemize}
	\item
		There is a cycle on the graph of $M$'s transition diagram,
		and a path from $q_0$ to a node on the cycle.
	\item
		The movements of the head-in-question add up to zero
		in a full round of that cycle.
	\item
		The readings of the head is consistent
		along the said path and cycle.
\end{itemize}
The definition of a head being windable
completely disregards the readings of the other heads,
hence the witness path and the cycle need not be a part of
a realistic execution of the machine $M$.

We will define the windable heads formally
to clarify its distinguishing points.
Some preliminary definitions will be needed.

\begin{definition}[Multi-step transition function]
	\label{defn:multistepdelta}
	\begin{equation*}
		\delta^\D \colon Q \times (\Gamma^\D)^k \to
		\P\left(\wo{Q}{q_0} \times (\Delta^\D)^k\right)
	\end{equation*}
	is the $\D$-step extension of the
	transition function $\delta$
	of a \tnfa{k} $M$.
	It is defined recursively, as follows:
	\begin{align*}
		\delta^1 &= \delta\\
		\delta^\D(q, g_1, \dotsc, g_k) &=
		\Set{
			(r, D_1 \ext d_1, \dotsc, D_k \ext d_k) |
			\begin{aligned}
				(r, d_1, \dotsc, d_k) &\in
				\delta(s, {g_1}_\D, \dotsc, {g_k}_\D) 
				\\ % \suchthat \\
				(s, D_1, \dotsc, D_k) &\in
				\delta^{\D-1}(q, \trim{g_1}, \dotsc, \trim{g_k})
			\end{aligned}
		}
	\end{align*}
\end{definition}

The set $\delta^\D(q, g_1, \dotsc, g_k)$ contains
a $(k+1)$-tuple for each nondeterministic computation to be performed by $M$,
as it starts from the state $q$ and reads $g_i$ with its $i$th head.
These tuples, each referred to as a \mydefn{computation log},
consist of the state reached,
and the movement histories of the $k$ heads during
that computation.

The constraint of a constant and persistent tape contents
that is present in an execution of a \tnfa{k} is blurred
in the definition for multi-step transition function.
This closely resembles the verifier's perspective
of the remaining heads that it does not verify
in the previous section. There, however,
the verifier's readings were consistent in itself.
This slight will be accounted for with the next pair
of definitions.

\begin{definition}[Relative head position during $i$th transition]
	\label{defn:headpos}
	The relative position of the head since the before
	the first movement in the movement history $D$ of length $t$
	and while making the $i$th transition of that history
	is given by the function
	$\POS_D(i) \colon \No^{\le \D} \to \openrange{-\D}{\D}$
	defined as
	$$
	\POS_D(i) = \SUM(D_{1 : i-1}).
	$$
\end{definition}

If $D$ is a movement history
from a computation that does not
attempt to move the head out of tape's bounds,
then $\POS_D(i)$ is the position of the head
while making the $i$th transition,
relative to the position where the head was
at the beginning of that computation.

\begin{definition}[1-head consistent $\delta^\D$]\label{defn:dk1}
	$\delta^\D_1 \colon Q \times (\Gamma^\D)^k \to
	\P\left(\wo{Q}{q_0} \times (\Delta^\D)^k\right)$ is the
	\mydefn{$i$th-head consistent} subset of
	$\delta^\D$ of a \tnfa{k} $M$.
	It filters out the first-head inconsistent computation logs by
	scrutinizing the purportedly read characters by examining
	the movement histories against the readings.
	The formal definition assumes that $M$ does not attempt to move
	its heads beyond the end markers, and is as follows:
	\begin{multline*}
		\delta^\D_1(q, g_1, \dotsc, g_k) = \left\{\:
			(r, D_1, \dotsc, D_k) \in \delta^\D(q, g_1, \dotsc, g_k)
			\:\right|\\
			\left.
			\forall p \in \openrange{-\D}{\D},\;
			\forall x, y \in \POS^{-1}_{D_i}(p)\;
			\BR{ {g_i}_x = {g_i}_y }
		\:\right\}
	\end{multline*}
\end{definition}

For each pair of transitions departing from the same tape cell,
it is checked whether the same symbol is read while being performed.
This check is needed to be done only for $p \in \openrange{-\D}{\D}$,
since in $\D$ steps,
a head may at most travel $\D$ cells afar,
and the last cell it can read from will then be the previous one.
This is also consistent with the definition of $\POS_D$.

This last definition is the exact analogue of the
verifiers' perspective in the algorithms proposed by~\cite{say_finite_2014}.
It can be used directly in our next definition,
that will lead us to a characterization of the \tnfa{k}.

\begin{definition}[Windable heads]\label{defn:wind}
	The $i$th head of a \tnfa{k} $M$ is \mydefn{windable}
	\textiff{} there exists;
	\begin{enumerate}[topsep=3pt, itemsep=1pt]
		\item
			$g_1, \dotsc, g_k \in \Gamma^\D$ and
			$g_1', \dotsc, g_k' \in \Gamma^l$,
			for $\D$ and $l$ positive,
		\item\label{itm:winddelta}
			$(q, D_1, \dotsc, D_k) \in \delta^\D_i(q_0, g_1, \dotsc, g_k)$,
		\item
			$(q, D_1 \ext D_1', \dotsc, D_k \ext D_k') \in
			\delta^{\D+l}_i(q_0, g_1 \ext g_1', \dotsc, g_k \ext g_k')$
			where $\SUM(D'_i) = 0$.
	\end{enumerate}
\end{definition}

When these conditions hold,
$g_1, \dotsc, g_k$ can be viewed as the sequences of characters
that can be fed to $\delta$ to bring $M$
from $q_0$ to $q$,
crucially without breaking consistency among the $i$th
head's readings.
This ensures reachability to state $q$.
Then, the sequences $g_1', \dotsc, g_k'$
wind the $i$th head into a loop;
bringing $M$ back to state $q$ and
the first head back to where it started the loop,
all while keeping the $i$th head's readings consistent.
The readings from the other heads are allowed to be inconsistent,
and their position can change with every such loop.

A head is \mydefn{reliable} \textiff{} the head is not windable.

It is important to note that a winding is not based on a
realistic execution of a \tnfa{k}.
A head of a \tnfa{k} $M$ might be windable,
even if it is guaranteed to halt on every input.
This is because the property of being windable
allows other heads to have \emph{unrealistic},
inconsistent readings that may be never realized with any input string.

\section{
	Recognizing some languages in \NL{} with
	constant-randomness and reducible-error verifiers
}\label{sec:alg}

Consider a language $L \in \NL$
with a \tnfa{k} recognizer $M$ that halts on every input.
In designing the randomness-restricted \tpfa1 verifier $V_2$,
following three cases will be considered:

\paragraph{All heads are reliable.}
In this case, $V_1$ suffices by itself
to attain reducible error.
Without any windable heads in the underlying \tnfa{k},
each round of $V_1$ will terminate.
The certificate can only make $V_1$ falsely accept,
and the chances for that can be reduced arbitrarily
by increasing $\C$.

\paragraph{All heads are windable.}
In this case, unless the worst-case assumptions are alleviated,
any verification algorithm using a simulation principle
similar to $V_1$ will be wound up on the first round.
The head with the minimum probability of getting chosen
will be the weakest link of $V_2$,
thus the head the certificate will be lying about.
The failure to reject rate is equal 1 minus that probability.
This rate is the lowest when the probabilities are equal,
and is then $\sfrac{k-1}k$.

\paragraph{It is a mix.}
Let $\kw, \kr$ denote the windable and reliable head counts,
respectively. Thus $\kw + \kr = k$.
The new verifier algorithm $V_2$ is similar to $V_1$,
but instead of choosing a head to simulate with equal probability,
it will do a \mydefn{biased branching}.
With biased branching, $V_2$ favors the
reliable heads over the windable heads while
choosing a head to verify.

Let $\opw, \opr$ denote the desired probability of choosing
a windable and reliable head, respectively.
Note that $\opw + \opr = 1$.
The probabilities of choosing a head within types
(windable or reliable) are kept equal.
Denote the probability of choosing a particular windable
head as $\pw = \sfrac{\opw}{\kw}$, and similarly
$p_r = \sfrac{\opr}{\kr}$.
Assume $\opw, \opr$ are finitely representable in binary,
and with $\B$ digits after the decimal point.
Then, the algorithm of $V_2$ is the same as $V_1$,
with the only difference at step~\ref{itm:randchoice}:
\begin{enumerate}
	\item[1c.${}'$]
		Randomly choose a head of $M$ by biased branching:
		\begin{itemize}
			\item
				Instead of flipping $\ceil{\log k}$ coins,
				flip $\B + \ceil{\log(\max(\kw, \kr))}$ coins.
				Let $z_1, z_2, \dotsc, z_{\B}$ be the outcomes of
				the first $\B$ coins.
			\item
				If $\sum_{i=1}^\B 2^{-i}z_i < \opw$,
				choose one of the windable heads
				depending on the outcomes of the next
				$\ceil{\log \kw}$ coins.
				Otherwise, similarly choose a reliable head
				via $\ceil{\log \kr}$ coins.
		\end{itemize}
\end{enumerate}

\paragraph{For an input string $w \in L$.}
Verifier $V_2$ is still perfectly accurate.
Certificate may provide any route that leads $M$ to acceptance.
Repeating this for $\C$-many times,
$V_2$ will accept after $\C$ rounds of validation.

\paragraph{For an input string $w \notin L$.}
To keep $V_2$ from rejecting,
the certificate will need to lie about
at least one of the heads.
Switching the head to lie about in between rounds
cannot be of any benefit to the certificate on its mission,
since the rounds are identical
both from $V_2$'s and the certificate's points of view.
Hence, it is reasonable to assume that
the certificate repeats itself in each round,
and simplify our analysis.

The worst-case assumption is that
the certificate can lie about a single (arbitrary) head alone
and deceive $V_2$ in the worst means possible,
depending on the head it chooses:
\begin{itemize}
	\item
		If it chooses the head being lied about,
		$V_2$ detects the lie rather than being deceived.
	\item
		Otherwise, if a windable head was chosen,
		$V_2$ loops indefinitely.
	\item
		Otherwise (i.e. a reliable head was chosen),
		$V_2$ runs for another round or accepts $w$.
\end{itemize}

The head which the certificate fixes to lie about is
either a windable head or a reliable one.
Given a $V_2$ algorithm with its parameter $\opw$ set,
let $\FRw(\opr)$ be the probability of $V_2$ failing to reject
against a certificate that lies about any one windable head.
Let $\FRr(\opr)$ similarly be the probability
for the reliable counterpart.

The most evil certificate would lie about the head that yields a
higher error. Thus, the worst-case failure to reject probability
is given by
\begin{equation*}
	\FR(\opr) = \max(\FRw(\opr), \FRr(\opr)).
\end{equation*}
Individually, $\FRw(\opr)$ and $\FRr(\opr)$
are calculated using the following formulae:

\begin{align*}
	\FRw(\opr) &= \sum_{i=0}^{\C-1} \opr^i (\opw-\pw) + \opr^{\C} \\
	&= (1 - \opr^{\C}) \cdot \left(1 - \frac1{\kw}\right) + \opr^{\C} \\
	&= 1 - \frac{1 - \opr^{\C}}{\kw}
\end{align*}
\begin{align*}
	\FRr(\opr) &= \sum_{i=0}^{\C-1} (\opr-\pr)^i \opw + (\opr-\pr)^{\C} \\
	&= \frac{1 - (\opr-\pr)^{\C}}{1 - (\opr-\pr)} \cdot \opw + (\opr-\pr)^{\C} \\
	&= \frac{\opw}{\opw+\pr} + \left(1-\frac{\opw}{\opw+\pr}\right)(\opr-\pr)^{\C}\\
\end{align*}

The objective is to find the optimum $\opr$,
denoted $\Propt$, minimizing the error $\FR(\opr)$.
We note that $\FR(1)$ is $1$.
Hence, $\Propt < 1$.

Constant $\C$ may be chosen arbitrarily large.
For $\opr < 1$, and $\C$ very large,
approximations of $\FRw$ and $\FRr$ are, respectively,
given as
\begin{align*}
	\FRw^*(\opr) &= 1 - \frac{1}{\kw} &
	\FRr^*(\opr) &= \frac{\opw}{\opw+\pr}.
\end{align*}

Error $\FRw^*$ is a constant between 0 and 1.
For $0 \le \opr \le 1$,
error $\FRr^*$ decreases from 1 to 0,
and in a strictly monotonous fashion:
\begin{equation*}
	\frac{\dd{\FRr^*}}{\dd{\opr}} =
	\frac{-\pr - \sfrac{\opw}{\kw}}{(\opw+\pr)^2} < 0
\end{equation*}
These indicate that $\FRw^*(\opr)$ and $\FRr^*(\opr)$ are equal
for a unique $\opr = \Propt$.
The optimality of $\Propt$ will be proved shortly.
It is easy to verify that
\begin{equation}
	\Propt = \frac{\kr}{k-1}.\label{eq:propt}
\end{equation}

Using $\Propt$ we can define $\FR^*$ as the following
partial function:
\begin{equation*}
	\FR^*(\opr) = \begin{cases}
		\FRr^*(\opr) & \text{for } \opr \le \Propt\\
		\FRw^*(\opr) & \text{for } \opr \ge \Propt
	\end{cases}
\end{equation*}

Since $\FRr^*$ is a decreasing function,
$\FR(\opr) > \FR(\Propt)$ for any $\opr < \Propt$.
The approximation $\FRw^*$ is a constant function.
Function $\FRw$, however, is actually an increasing one.
Therefore, given $\C$ large,
probability $\Propt$ approximates the optimum
for $V_2$ choosing a reliable head among the $k$
heads of the $M$,
while verifying for the language $\L{M} \in \NL$.
Consequently the optimum error for $V_2$ is
\begin{equation}
	\FR(\Propt) = 1 - \frac1{\kw}.\label{eq:fropt}
\end{equation}

This points to some important facts.

\begin{theorem}\label{lem:kwdefineserr}
	The minimum error for $V_2$ depends only on the number
	of windable heads of the \tnfa{k} $M$ recognizing $L \in \NL$.
\end{theorem}

\begin{definition}[Reducible strong error subset of \NL{}]
	For $\err > 0$ approaching zero, 
	the reducible strong error subset of \NL{} is defined as
	\begin{equation*}
		\rNL{} = \NL{} \cap \Ls{\err}{\owips{\tpfa{1}, \constr}}.
	\end{equation*}
\end{definition}

\begin{theorem}\label{lem:kwzeroorone}
	For $\kw \le 1$ and $\kw$ arbitrary,
	\begin{equation*}
		\L{\tnfa{\kw+\kr}} \subseteq \rNL{}.
	\end{equation*}
\end{theorem}

\Cref{eq:propt,eq:fropt}, and their consequent
\cref{lem:kwdefineserr,lem:kwzeroorone},
constitute the main results of this study.

Similar to how $V_1'$ was obtained,
the algorithm for $V_2'$ is as follows:
\begin{enumerate}
	\item
		Randomly \reject{} with $\sfrac{\kw-1}{2\kw}$ probability
		by flipping $\ceil{\log \kw} + 1$ coins.
	\item
		Continue as $V_2$.
\end{enumerate}
The strong error of $V_2'$ is then given by
$\serr(V_2') \le \sfrac12 - \sfrac1{2\kw}$.

\subsection{Example languages from \rNL{} and potential outsiders}

Let $w_{\ta}$ denote the amount of symbols $\ta$ in a string $w$.

Following two are some
example languages with \tnfa{\kw+\kr} recognizers, where $\kw = 0$:
\begin{align*}
	A_1
	&= \Set{ \ta^n\tb^n\tc^n\td^n | n \ge 0 }\\
	A_2
	&= \Set{ w \in \Set{ \ta, \tb, \tc } | w_{\ta} = w_{\tb} = w_{\tc} }\\
	\intertext{
		An example language with a $\kw \le 1$ recognizer is the following:
	}
	A_3
	&= \Set{ \ta_1\ta_2 \dotsm \ta_n \tsh \ta_1^+\ta_2^+ \dotsm \ta_n^+ | n \ge 0}
	\intertext{
		Lastly, an example language that might be outside \rNL{}
		is follows:
	}
	A_4
	&= \Set{ w \in \Set{ \ta, \tb, \tc } | w_{\ta} \cdot w_{\tb} = w_{\tc} }
\end{align*}

\section{Open Questions}\label{sec:open}

It is curious to us whether
$\L{\tnfa{\kw+\kr}}$
coincides with any known class of languages
for $\kw = 0$ or $1$, or $\kw \le 1$.
The minimum number or windable heads required
for a language in \NL{} to be recognized
by a halting \tnfa{k}, could establish a complexity class.
Conversely, one might be able to discover yet another
infinite hierarchy of languages based on the number of
windable heads.
For some $c > 0$ and $\kw' = \kw + c$,
this hierarchy might be of the form
\begin{equation*}
	\L{\tnfa{k = \kw+\kr}} \subsetneq \L{\tnfa{k' = \kw'+\kr'}}
\end{equation*}
for $k = k'$, $\kr = \kr'$, or without any further restriction.

\begingroup
\raggedright
\bibliographystyle{splncs04}
\bibliography{ref} 
\endgroup

\end{document}